\documentclass[sigconf]{acmart}
\usepackage{xcolor}

\definecolor{lightgray}{gray}{0.9}
\usepackage[many]{tcolorbox}

\usepackage{enumerate}
\usepackage[shortlabels]{enumitem}

\usepackage{xspace}
\usepackage{enumitem}
\usepackage{algorithm}
\usepackage{algpseudocode}
\usepackage{hyperref}
\usepackage{cleveref}
\newtheorem{theorem}{Theorem}

\newtheorem{definition}{Definition}
\usepackage{pdfcomment}
\usepackage{multicol}
\usepackage{multirow}

\newcommand{\framework}{Femur\xspace}
\newcommand{\pgm}{PGM-index\xspace}

\newtcbtheorem[]{point}{}%
    {enhanced,
    colback = black!10, 
    colbacktitle = black!10,
    coltitle = black,
    attach title to upper,
    after title={:\ },
    boxrule = 0pt,
    frame hidden,
    fonttitle = \bfseries\sffamily,
    breakable,
    before skip = 2ex,
    after skip = 1ex,
    boxsep = 1.25ex,
    top = 0ex,
    bottom = 0ex,
    left = 0ex,
    right = 1.5ex,
}{point}

\AtBeginDocument{%
  }

\setcopyright{acmlicensed}
\copyrightyear{2025}
\acmYear{2025}
\acmDOI{XXXXXXX.XXXXXXX}

\acmConference{2025 ACM SIGMOD/PODS Conferences}{June 22-27}{Berlin, Germany}

\begin{document}

\setlength{\textfloatsep}{3mm}

\title{\framework: A Flexible Framework for Fast and Secure Querying from Public Key-Value Store}


\author{Jiaoyi Zhang}
\email{jy-zhang20@mails.tsinghua.edu.cn}
\orcid{0009-0008-3075-6147}
\affiliation{
  \institution{Tsinghua University}
  \city{Beijing}
  \country{China}
}

\author{Liqiang Peng}
\email{plq270998@alibaba-inc.com}
\affiliation{
  \institution{Alibaba Group}
  \city{Beijing}
  \country{China}
} 

\author{Mo Sha}
\email{shamo.sm@alibaba-inc.com}
\affiliation{
  \institution{Alibaba Cloud}
  \country{Singapore}
} 

\author{Weiran Liu}
\email{weiran.lwr@alibaba-inc.com}
\affiliation{
  \institution{Alibaba Group}
  \city{Beijing}
  \country{China}
} 

\author{Xiang Li}
\email{lixiang20@mails.tsinghua.edu.cn}
\affiliation{
  \institution{Tsinghua University}
  \city{Beijing}
  \country{China}
} 

\author{Sheng Wang}
\email{sh.wang@alibaba-inc.com}
\affiliation{
  \institution{Alibaba Cloud}
  \country{Singapore}
} 

\author{Feifei Li}
\email{lifeifei@alibaba-inc.com}
\affiliation{
  \institution{Alibaba Cloud}
  \city{Hangzhou}
  \country{China}
} 

\author{Mingyu Gao}
\authornotemark[1]
\email{gaomy@tsinghua.edu.cn}
\affiliation{
  \institution{Tsinghua University}
  \city{Beijing}
  \country{China}
}

\author{Huanchen Zhang}
\authornote{Huanchen Zhang and Mingyu Gao are also affiliated with Shanghai Qi Zhi Institute. Corresponding authors.}
\email{huanchen@tsinghua.edu.cn}
\orcid{0009-0001-4821-1558}
\affiliation{
  \institution{Tsinghua University}
  \city{Beijing}
  \country{China}
}

\renewcommand{\shortauthors}{Jiaoyi Zhang et al.}
\newcommand{\para}[1]{\vspace{1mm}\noindent\textbf{{#1}}\xspace}


\begin{abstract}

With increasing demands for privacy, it becomes necessary to protect sensitive user query data when accessing public key-value databases. 
Existing Private Information Retrieval (PIR) schemes provide full security 
but suffer from poor scalability, limiting their applicability in large-scale deployment.
We argue that in many real-world scenarios, a more practical solution should allow users to flexibly determine the privacy levels of their queries in a theoretically guided way, balancing security and performance based on specific needs.
To formally provide provable guarantees, we introduce a novel concept of distance-based indistinguishability, which can facilitate users to comfortably relax their security requirements. 
We then design \framework, an efficient framework to securely query public key-value stores with flexible security and performance trade-offs. 
It uses a space-efficient learned index to convert query keys into storage locations, obfuscates these locations with extra noise provably derived by the distance-based indistinguishability theory, and sends the expanded range to the server. 
The server then adaptively utilizes the best scheme to retrieve data. We also propose a novel variable-range PIR scheme optimized for bandwidth-constrained environments. 
Experiments show that \framework outperforms the state-of-the-art designs even when ensuring the same full security level. When users are willing to relax their privacy requirements, \framework can further improve the performance gains to up to 163.9$\times$, demonstrating an effective trade-off between security and performance.



\end{abstract}

\maketitle
\setcounter{page}{1}
\section{Introduction}
The growing data volumes and pervasive cloud services have spurred efforts to protect user data privacy. Protecting user queries to public datasets is as critical as protecting the underlying database. For example, users checking phone numbers against a public scam database are vested in keeping the numbers private to the hosting server. Similarly, users querying sensitive health or financial data risk potential discrimination (e.g., biased treatment in insurance or employment) if their queries are revealed. These use cases highlight the need for secure and efficient query mechanisms for accessing public key-value stores while keeping queries confidential.

Private Information Retrieval (PIR) schemes~\cite{chor1997private,menon2022spiral,mughees2021onionpir,melchor2016xpir,corrigan2020private,kogan2021private,henzinger2023one} have emerged to address the above concerns.
They allow users to upload encrypted keys for confidential server-side computation and results return without decryption. The entire process ensures that sensitive information remains secure.
Despite advancements in PIR techniques---such as improvements in cryptographic protocols~\cite{davidson2023frodopir,gentry2009fully,brakerski2014leveled,mahdavi2022constant,ahmad2022pantheon}, encoding strategies~\cite{angel2018pir}, and pre-processing optimizations~\cite{davidson2023frodopir,chen2017simple,celi2024call,zhou2024piano}---they still face challenges with scaling. To guarantee query path privacy, most designs must process all the $n$ key-value pairs in the database with up to $O(n)$ complexity, leading to high query response time in large-scale applications. For example, Pantheon~\cite{ahmad2022pantheon} takes 1.15 seconds to execute a query on only 65,536 key-value pairs.

In this paper, we propose a more practical solution that relaxes security while upholding theoretical guarantees. Existing PIR schemes yield impractical processing times for large real-world datasets, with delays ranging from seconds to hours (Pantheon~\cite{ahmad2022pantheon} takes over 2 hours on $2^{24}$ records). Moreover, full dataset obfuscation may not be necessary in many applications. For example, an American user may only need to hide their queries within the U.S. phone numbers rather than the global dataset. In fact, industry practices often relax security by using hash functions~\cite{heinrich2021privatedrop} to partition datasets into smaller buckets to improve query performance by reducing the number of key-value pairs involved per query. However, this method lacks formal security guarantees: the bucket IDs are exposed during computation, and an uneven hash partition could produce single-pair buckets, thus revealing the user's query.

We introduce \framework, a framework enabling users to control and balance performance and privacy with theoretical security guarantees. \framework allows users to configure from no security to full security backed by our formal definition of ``relaxed security'' (i.e., \emph{distance-based indistinguishability}). Depending on the chosen level, \framework's user-side module automatically sends an obfuscated range of keys (including the user's real key) to the server, ensuring that the real key is indistinguishable under the relaxed security level without unnecessary performance loss. When this key range extends to the entire database, \framework achieves full privacy as before.

\framework consists of three main components: key-to-position conversion, obfuscated range generation, and adaptive key-value retrieval. First, the user-side \emph{key-to-position conversion} module maps query keys to their storage locations within the database, similar to existing keyword PIR schemes. The novelty lies in that we use the PGM-index~\cite{ferragina2020pgm}, a state-of-the-art, space-efficient learned index for the mapping. The prediction error of \pgm can be handled together with the obfuscated query range. Learned indexes are typically faster and smaller than traditional indexes (e.g., B+trees) and can therefore speed up the conversion while significantly reducing the index size stored on the client side.

Second, the \emph{obfuscated range generation} module adds noise to the range predicted by the \pgm before sending it to the server so that the server-side computation is confined to the key-value pairs within the range. This obfuscated range satisfies distance-based indistinguishability with a user-specified security level guaranteed by our theory. The higher the security level, the slower the query.

Third, the server-side \emph{adaptive key-value retrieval} module employs a cost model to choose the best strategy, trading between computation and network bandwidth. For high-bandwidth networks, the client side downloads all key-value pairs in the obfuscated range (of length $s$) in plaintext. This approach conceals the query key without homomorphic computation but incurs $O(s)$ transmission complexity. When the bandwidth is limited, we propose a novel variable-range PIR scheme with a transmission complexity of $O(1)$ but a computation complexity of $O(s)$.

Our experiments show that \framework scales well on a dataset containing 200 million records (a size infeasible for previous solutions) and outperforms the two baselines (Chalamet~\cite{celi2024call} and Pantheon~\cite{ahmad2022pantheon}) by 1.05$\times$ and 7.71$\times$, respectively under full security, and by 163.9$\times$ and 1206.1$\times$, respectively with a relaxed security guarantee. Additionally, \framework's offline initialization phase only takes a few minutes, compared to several hours from previous designs. We also integrated \framework into Redis~\cite{redis}, a popular in-memory key-value store, and demonstrated that \framework also supports efficient value updates in real-world scenarios.

We make the following contributions.
\vspace{-.2em}
\begin{itemize}[leftmargin=*]
    \setlength\itemsep{.5mm}
    \setlength\topsep{.5mm}
    \item We identify that full security prevents current PIR schemes from scaling and is often unnecessary in real-world applications.

    \item We introduce relaxed security by formally defining ``distance-based indistinguishability'' to allow users to trade between performance and privacy while offering provable security guarantees.
    
    \item We propose \framework, a user-centric framework that eliminates unnecessary computational and communication costs by incorporating a learned index with a novel variable-range PIR scheme.

    \item We demonstrate \framework's superior performance and scalability over the baselines via a thorough evaluation.
\end{itemize}

\section{Preliminaries}

\subsection{Private Retrieval From Public Data}\label{sec:pir}
This paper focuses on user privacy in public key-value stores by ensuring secure lookups that hide both the users' query keys and access patterns. This is often realized by traditional PIR schemes, including keyword PIRs and index PIRs. Depending on whether pre-processing is allowed, PIR schemes can be further categorized as stateless PIR (without pre-processing )~\cite{angel2018pir,melchor2016xpir,ahmad2021addra,ahmad2022pantheon} often with linear computational complexity, and as stateful PIRs~\cite{corrigan2020private,kogan2021private,henzinger2023one,davidson2023frodopir,zhou2024piano,celi2024call} which leverage offline pre-processing to reduce online computation.

\subsubsection{Index-PIR}
Index PIRs assume that the user knows the exact location of the data to be queried. 
Typically, the user encodes the location as a one-hot vector, encrypts it, and sends it to one or more untrusted servers for private retrieval. 
The server then performs a privacy-preserving inner product between the encrypted query and its dataset, returning the encrypted result to the user. 
We focus on the single-server scenario, as the multi-server setup relies on the less practical assumption of non-colluding servers.

Existing index-PIR schemes like SimplePIR~\cite{henzinger2023one} and FrodoPIR~\cite{davidson2023frodopir} reduce online computation via one-time pre-computation, where users download specific ``hints'' in the offline phase.
However, their online phase still scales linearly with dataset size. Piano~\cite{zhou2024piano} takes a different approach where the user continuously uploads sets of multiple locations during pre-processing and obtains the XOR of these data as hints. 
The online query must match one of these hints, accelerating queries but incurring offline communication overhead equal to the dataset size, rendering it impractical.
Since offline interactions between the user and the server are inevitable (e.g., establishing a communication channel), \framework also allows transferring hints during the offline phase. However, it ensures that the hint size remains significantly smaller than the dataset. To achieve this, \framework introduces a novel design that leverages learned indexes for faster pre-processing and efficient online processing.

Additionally, optimizations in cryptographic primitives~\cite{mughees2021onionpir,menon2022spiral} and batch processing~\cite{yeo2023lower,mughees2023vectorized} are orthogonal to our design and could enhance our framework with minor adjustments.

\subsubsection{Keyword-PIR}
In real-world scenarios, users often do not know the exact location of the data they wish to retrieve. Keyword PIRs address this limitation by allowing queries based on keywords and are typically implemented on top of index PIRs. The earliest approach~\cite{chor1997private} achieves this by mapping keywords to data locations through logarithmic rounds of communication. However, this introduces substantial overhead, rendering the method inefficient. 
For example, retrieving data from one million key-value pairs may require up to 21 round trips, resulting in significant overhead.

Recent keyword PIR schemes reduce rounds to one. Approaches based on fully homomorphic encryption (FHE)~\cite{gentry2009fully,brakerski2014leveled} achieve this by using equality operators~\cite{ahmad2022pantheon,mahdavi2022constant}. In these schemes, the user encrypts the desired keyword and uploads it to the server, which performs an equality check between the encrypted keyword and all keys in the dataset. This results in an encrypted one-hot vector that can be used for consequent index PIR. While these methods reduce communication, equality checks remain computationally expensive. 
To further enhance efficiency, Chalamet~\cite{celi2024call} employs techniques such as cuckoo hashing or filters to map keys to multiple potential locations, using index PIR to retrieve and combine these values. 
Compared to FHE-based methods, Chalamet reduces the online server-side computations, but still incurs high communication overhead and requires substantial pre-processing.

These schemes improve retrieval efficiency while maintaining full security. However, \framework provides extra flexibility by allowing users to adjust security levels according to their specific needs, achieving a better balance between performance and privacy.

\subsection{Differential Privacy}\label{sec:bg_dp}
Differential privacy (DP) is a rigorous mathematical method widely used in database systems, aimed at preventing adversaries from inferring individual data from query results by adding controlled noise.
Classical DP applications assume the presence of a trusted party that collects all the data and adds noise to the entire dataset, an approach adopted by various organizations including Google~\cite{amin2022plume} and Uber~\cite{johnson2018towards,johnson2020chorus}.
In contrast, Local Differential Privacy (LDP) removes the need for a trusted party by allowing users to obfuscate data locally before sharing, as used by Google~\cite{erlingsson2014rappor,aktay2020google}, Apple~\cite{tang2017privacy}, and Meta~\cite{messing2020facebook}.
Besides, users have the flexibility to choose their desired privacy level based on the sensitivity of their information~\cite{alvim2018local}. 

\begin{definition}[Local Differential Privacy, LDP]
An obfuscation mechanism $\mathcal{M}:\mathbb{D}\rightarrow \mathbb{O}$ satisfies $\epsilon$-LDP with privacy level $\epsilon$ ($\epsilon\geq 0$) if for any $x, x'\in \mathbb{D}$ and any output $y\in \mathbb{O}$, we have
\begin{equation}
    \Pr[\mathcal{M}(x)=y]\leq e^\epsilon\cdot \Pr[\mathcal{M}(x')=y]
\end{equation}
\end{definition}
The private parameter $\epsilon$ represents the privacy level, with larger $\epsilon$ indicating weaker privacy protection provided by 
$\mathcal{M}$. To satisfy the condition of LDP, $\mathcal{M}$ must effectively hide all differences between data entries, including the difference between the maximum and minimum values. In practice, this often requires adding a significant amount of noise, leading to reduced usability.

\begin{definition}[Distance-based Local Differential Privacy]
An obfuscation mechanism $\mathcal{M}:\mathbb{D}\rightarrow \mathbb{O}$ satisfies $\epsilon$-dLDP with privacy level $\epsilon$ ($\epsilon\geq 0$) if for any $x, x'\in \mathbb{D}$ such that $\left|x-x'\right|\leq t$, and for any possible output $y\in \mathbb{O}$, we have
\begin{equation}
    \Pr[\mathcal{M}(x)=y]\leq e^{t\epsilon}\cdot \Pr[\mathcal{M}(x')=y]
\end{equation}
\end{definition}

According to this definition, the indistinguishability between any two sensitive data points decreases as the distance $t$ between them increases~\cite{chatzikokolakis2013broadening}. This strikes a balance between privacy and utility and is more suitable for practical application scenarios.

\begin{theorem}[Post-Processing~\cite{dwork2014algorithmic}]\label{theo:post-process} Let  $\mathcal{M}:\mathbb{D}\rightarrow \mathbb{O}$ be an $\epsilon$-dLDP obfuscation mechanism, and $f: \mathbb{O} \rightarrow \mathbb{O}'$ be any randomized function. Then, $f \circ \mathcal{M}$ remains $\epsilon$-dLDP. \end{theorem}

Immunity to post-processing is a key property of differential privacy, meaning that arbitrary transformations can be performed on the output without compromising privacy guarantees.

\subsection{Learned Indexes}
Similar to PIR schemes that employ filters for keyword PIR, index structures (e.g., B+tree~\cite{bayer1972symmetric} or learned indexes) can also be used to map keywords to their corresponding locations before retrieval. Learned indexes, first introduced by Kraska et al.~\cite{kraska2018case}, utilize data distribution characteristics to build more efficient index structures. Specifically, it focuses on data distribution and rank to map a given key to its corresponding memory address~\cite{wongkham2022updatable,lan2023updatable,zhang2022carmi,ding2020alex,wu2021updatable,zhang2024making}. This operation can be understood as a cumulative distribution function (CDF) applied to the key distribution.

Most learned indexes follow a similar structure and search process. Given a set of key-value pairs, they use a construction algorithm to create a tree structure, typically consisting of a root node, one or more levels of internal nodes, and leaf nodes that manage the key-value pairs. Each node contains a simple linear model and a small amount of metadata representing the distribution of a subset of the entire dataset.
To query a key, the learned index performs a top-down traversal. At each level, a model is used to predict the next node's location to be accessed at the next level, continuing until a leaf node is reached. The model in the leaf node predicts the position (i.e., rank) of the query key, and a last-mile search is performed within a small range to determine the exact key location.

\section{Motivation}\label{sec:motiv}
We focus on a scenario where a server hosts a public key-value store, allowing users to retrieve values by their keys without revealing which specific keys they are querying. 
Such scenarios are common, including untraceable browsing~\cite{kogan2021private,wu2016privacy,henzinger2023private}, contact discovery~\cite{borisov2015dp5,demmler2018pir}, password leakage detection~\cite{thomas2019protecting}, anonymous messaging~\cite{angel2018pir,angel2016unobservable,mittal2011pir,kwon2015riffle}, etc. 
For example, at WWDC 2024, Apple introduced a live caller ID lookup feature that enables users to obtain information about incoming calls from a public dataset without revealing the queried phone number to the server. This feature helps safely block nuisance calls, offering both security and convenience.

In such scenarios, typically the server is the data owner. The users can privately perform read-only lookups to the public key-value store, but usually not allowed to modify the data. Only the server may periodically update its data as needed. 
This is fundamentally different from the cases where the users outsource their own databases to untrusted servers. But we emphasize that it is still a common and important problem in real world, as exemplified above. 
For example, an investor may wish to privately retrieve information about a particular stock, such as its trading record, without intending to update any public data. However, these queries often involve sensitive information (e.g., investment interests), necessitating privacy guarantees.

Most existing approaches rely on PIR schemes to support such scenarios. 
As discussed in~\Cref{sec:pir}, PIR usually uses an encrypted one-hot bitmap to represent the location of the desired key-value pair, and performs homomorphic inner products with all key-value pairs in the dataset, thereby concealing the queried key. However, their computational cost scales with the size of the database, making them impractical for large datasets. 
For example, Pantheon~\cite{ahmad2022pantheon} takes several hours to execute a query on a dataset containing 16,777,216 key-value pairs. 
Although Pantheon can parallelize computations by distributing the dataset across multiple machines, this does not fundamentally reduce the computational complexity of each query, leaving the scalability problem unresolved.

In this paper, we argue that such impractical performance of existing schemes stems from their rigid guarantees of full security, which require all key-value pairs to be involved in computation or transmission. Current PIR schemes cannot bypass this requirement, even with optimized cryptographic algorithms or preprocessing. 
However, this requirement is not always necessary in practical applications. In real-world scenarios, users' security requirements are independent of the dataset size and often remain fixed. 
For example, when protecting a user's address, the user may be comfortable revealing her country (e.g., the United States) but not her city (e.g., San Francisco). This allows the query to be executed securely over all addresses within the country, rather than over the entire global dataset. 
Moreover, in a geographic database like Open Street Map, which is keyed by latitude and longitude, nearby data items in the database are geographically close to each other. 
Practical applications often involve scenarios where the storage location in the database is correlated with the value of the key. 
Meanwhile, such ``weaker'' security is common in industrial solutions, which often use hash functions to partition the complete dataset into smaller subsets. 
However, such methods not only leak the bucket ID where the queried key is located but also fail to provide provable privacy guarantees in general scenarios.
For example, if a bucket contains only a single entry, the queried key is immediately exposed.

To address these issues, we believe it is essential to offer flexible degrees of relaxed security based on the user's specific needs. This flexibility should still offer provable theoretical guarantees at all security levels. With a solid theoretical foundation, users can feel comfortable using more practical solutions that are not fully secure. 
To achieve this goal, the first requirement is to formally define how to relax security levels with theoretical guarantees. 
Then, a practical system design is needed to realize flexible privacy and ensure efficient query processing at all security levels. 
In addition, it should not impose significant computational or storage burdens on users, meaning offline computation and storage should be minimized.

\section{Relaxed Security for Private Retrieval}\label{sec:relaxed}
Our objective is to develop a flexible and scalable framework that enables users to privately retrieve the value associated with a key from a public key-value store. The primary security goal is to hide the user's lookup queries from the server, while the server's key-value data are public.
Our scenario resembles that of the keyword PIR~\cite{angel2018pir,ahmad2022pantheon,celi2024call,mahdavi2022constant}, with the key distinction that it permits the flexibility of enforcing different levels of relaxed security based on specific user requirements. 
In this section, we formalize our novel definition of relaxed security and the corresponding threat model.

\subsection{Problem Formulation}\label{sec:formulation}

Let $DB = \{(k_0, v_0), (k_1, v_1), \ldots, (k_{n-1}, v_{n-1})\}$ denote the set of public key-value pairs held by the server, where the key set $K = \{k_0, k_1, \ldots, k_{n-1}\}$ serves as the primary key of the database (no duplication). $DB$ is sorted by primary keys. 
Each key-value pair is stored in plaintext with a uniform length to prevent attackers from inferring the queried key based on length variations.
Let $k_{\text{target}}$ be the querying key of the user, and $v_{\text{target}}$ be the corresponding value. We discuss our problem below in terms of correctness and privacy.

\subsubsection{Correctness}
If $k_{\text{target}} \in K$, the server must return the corresponding $v_{\text{target}}$; otherwise, it returns null, allowing the client to readily verify that $k_{\text{target}}$ does not exist.

\subsubsection{Privacy: distance-based indistinguishability}\label{sec:privacy}
Our privacy definition aims to protect query privacy within a distance less than $t$ in the database $DB$. 
Here, distance refers to the difference in storage positions of two querying keys within the sorted key-value store, with $t$ representing the maximum allowable distance specified by the client.
Let $q_i$ and $q_j$ represent queries corresponding to keys $k_i$ and $k_j$, respectively (possibly chosen by the adversary), where $|i - j| \leq t$.
The user selects one query to execute, while the adversary observes the resulting events, denoted by $O_i$ and $O_j$ (e.g., the key-value pairs involved in the computation), and attempts to distinguish which query was executed.

Inspired by LDP introduced in~\Cref{sec:bg_dp}, we define a scheme as achieving distance-based indistinguishability (denoted as $\epsilon$-dist indistinguishability) 
for a given maximum allowable distance $t$ ($t>0$), 
if there exists a non-negative constant $\epsilon$ such that for any pair of queries $q_i$ and $q_j$ where $|i - j| \leq t$ ($i, j \in [0, n)$), and for all possible adversarial observations $O$ in the observation space $\Omega$, the following condition holds:
\begin{equation}\label{equ:definition}
\Pr[O \mid q_i] \leq e^{\epsilon} \cdot \Pr[O \mid q_j]
\end{equation}
where $\Pr[O \mid q_i]$ represents the probability of the observation $O$ given that $q_i$ has been issued.
This ensures that the adversary cannot distinguish between queries $q_i$ and $q_j$ based on observed events, thereby maintaining distance-based indistinguishability. 
The parameters $\epsilon$ and $t$ together determine the security guarantees, where $\epsilon$ is the privacy parameter and $t$ is the maximum allowable distance between two queries for the indistinguishability guarantee to hold. When $t$ is fixed, increasing $\epsilon$ weakens the indistinguishability, resulting in more relaxed security guarantees. Conversely, when $\epsilon$ is fixed, a larger $t$ strengthens the security guarantees. In \framework, the ratio of these two parameters determines the expected number of data points within the obfuscation range, as discussed in~\Cref{sec:dp}. 

To minimize user burden, we set the default value of $\epsilon$ to $2^{-6}$, since $t$ is generally more intuitive for users to adjust.
This empirical default value of $\epsilon$, chosen as a power of 2, aligns with the highest security levels in prior DP-related work~\cite{li2022opboost}.
Users can adjust $t$ to specify how many neighboring queries should be included to make the real query indistinguishable among them, allowing for a trade-off between privacy and performance based on individual needs. 
For example, consider a dataset of 100 queries and $DB = \{(k_0, v_0), (k_1, v_1), \ldots, (k_{99}, v_{99})\}$. Suppose the user queries $k_i = k_{20}$ and sets $t$ to 10. In this case, any potential output corresponding to a query $q_j$ for any key within the range $[k_{10}, k_{30}]$ and the query $q_{20}$ for $k_{20}$ satisfies~\Cref{equ:definition}. This ensures that adversaries cannot determine which specific key the user is querying within $[k_{10}, k_{30}]$, even when observing the access pattern.

\begin{figure*}[t]
    \centering
    \includegraphics[width=0.88\linewidth]{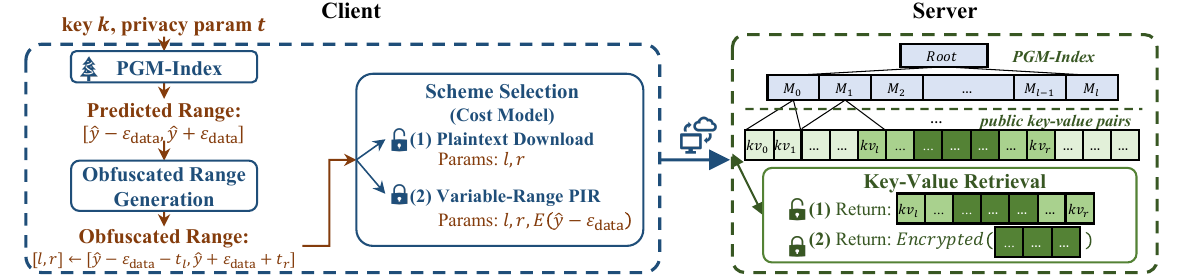}
    \caption{The Core Components of \framework Framework.}
    \label{fig:femur}
\end{figure*}

Intuitively, we can obfuscate the real query by adding fake queries. Let $S_i$ represent the set of queries that includes $q_i$ and extra fake queries generated by the client through a perturbation mechanism that satisfies $\epsilon$-dist indistinguishability, where $|S_i|$ is the number of queries. 
Then, the probability of correctly guessing $q_i$ is $\frac{1}{|S_i|}$. If $|S_i| = n$, full security and maximum uncertainty are achieved. If $|S_i| < n$, it corresponds to relaxed quantity-based agnosticism, reducing the uncertainty. 
Finally, if $|S_i| = 1$, no security is provided, as the adversary can directly infer the query.

\subsection{Threat Model}\label{sec:threat_model}
The user locally owns a trusted client machine. It aims to query some data in the public key-value store on an untrusted server. The stored data is in plaintext. A semi-honest adversary may control the server and want to steal sensitive information about which item is returned. The adversary can monitor incoming queries from users, observe server responses, and see any operations performed by the server, including which data are involved in the computations (i.e., the data access patterns). It can measure the execution time, which depends on the number of queries but is not affected by which specific keys are queried.
It can also monitor, record, and analyze data communication between the client and the server, even utilizing any prior knowledge and observations, as well as the query history. 
However, the adversary cannot break cryptographic schemes like AES. This ensures that the adversary cannot decrypt the encrypted data (i.e., the querying key and the returned result). And it will not tamper with any data or code executed on the server. 
\section{Design Overview}
\framework allows users to flexibly control the privacy-preserving granularity of queries for efficiently processing. 
\framework has two phases: offline initialization and online query phase. 
The initialization phase is executed offline only once when there is no database update, allowing users to gather essential information from the server.
The online query phase includes the entire process from client-side query generation to receiving results from the server.
In this section, we outline the key components and workflow of \framework.

\subsection{Offline Initialization Phase}
When a client arrives, it begins by exchanging one-time information with the server during the initialization phase. Specifically, the server provides the client with essential details about the database, including the number of key-value pairs, the bit length of each pair, and the available retrieval schemes with their corresponding parameters. 
Besides, the server supplies an auxiliary index to allow the client to locally convert querying keys into (approximate) server storage locations. 
\framework offers two optional retrieval schemes: \emph{plaintext download} with no extra parameters, and \emph{variable-range PIR}, which requires three cryptographic parameters: $N$, $p$, and $q$ (as detailed in~\Cref{sec:variablePIR}). 
The client then constructs the public and private cryptographic keys needed to encrypt/decrypt queries in the variable-range PIR scheme. 
Upon completing this step, the client sends its cryptographic public key to the server. The server retains each client's public key for future interactions.

The auxiliary index structure provided by the server is the \pgm, a space-efficient learned index that does not store the actual key values but only the model parameters (e.g., slopes, intercepts, and partitioning keys).
This reduces the data transmitted to clients to tens of KB to a few MB, greatly cutting initialization overhead.

If the key-value store is updated (usually periodically by the server; \Cref{sec:motiv}), the initialization phase needs to be re-executed. Particularly, the new auxiliary index structure needs to be synchronized with the clients. \Cref{sec:update} further discusses the details.

\subsection{Online Query Phase}
The online phase includes the entire query lookup process. After determining the querying key $k$ and the desired privacy level $t$, the client generates a secure query range utilizing three components provided by \framework: key-to-position conversion, obfuscated range generation, and scheme selection. The obfuscated range boundaries and the selected retrieval scheme are then sent to the server. The server processes the query using the specified scheme and returns the corresponding key-value pairs. Finally, the client retrieves the desired value by performing a simple verification. Below we briefly describe the main building blocks of \framework, as shown in~\Cref{fig:femur}.

\subsubsection{Key-to-Position Conversion}
We employ the \pgm, a learned index structure, to map the user's querying key to its possible locations within the dataset (\Cref{sec:index}). It guarantees that the querying key will be found within these positions unless it does not exist in the database. The parameter $\varepsilon_{\text{data}}$ in the \pgm determines the range of possible positions. For a given key $k$, the \pgm outputs a predicted position $\hat{y}$, and the possible locations are then bounded by $[\hat{y} - \varepsilon_{\text{data}}, \hat{y} + \varepsilon_{\text{data}}]$.

\subsubsection{Obfuscated Range Generation}
In this step, we efficiently convert the range $[\hat{y} - \varepsilon_{\text{data}}, \hat{y} + \varepsilon_{\text{data}}]$ into a obfuscated range $[l, r]=[\hat{y} - \varepsilon_{\text{data}}-t_l, \hat{y} + \varepsilon_{\text{data}}+t_r]$. 
Specifically, this obfuscated range is generated using a noise generation mechanism that satisfies distance-based indistinguishability (\Cref{sec:dp}). It ensures that the server cannot infer the real key being queried, while eliminating the need for all data points to be involved in the computation.

\subsubsection{Scheme Selection}
Once the obfuscated range is determined, our framework uses a cost model to select the most efficient retrieval scheme for the current query between plaintext downloads and variable-range PIRs (\Cref{sec:retrieval}). 
The cost model is lightweight and can be executed efficiently on the client side, leading to minimal overhead.
For plaintext downloads, only the unencrypted boundaries $l$ and $r$ are sent to the server. 
For variable-range PIRs, besides $l$ and $r$, the left boundary ($\hat{y} - \varepsilon_{\text{data}}$) of the predicted range is encrypted and also sent to the server. 

\subsubsection{Server-Side Query Processing}
Upon receiving the query request, the server processes the query using the designated retrieval scheme (\Cref{sec:retrieval}). 
For plaintext downloads, the server sends the key-value pairs in the range $[l, r]$ directly back to the client in unencrypted form. For variable-range PIR, the server performs homomorphic encryption computations and returns a single ciphertext containing all key-value pairs in $[\hat{y} - \varepsilon_{\text{data}}, \hat{y} + \varepsilon_{\text{data}}]$. 
We carefully control $\varepsilon_{\text{data}}$ and the plaintext encoding scheme to guarantee the queried pair is always returned. 
Note that the server is unaware of the specific granularity of the querying key or the user's privacy settings throughout the retrieval and transmission process.

\section{Key-to-Position Conversion}\label{sec:index}
To facilitate privacy-preserving data retrieval on the server side, the querying key needs to be converted into the corresponding position of the key-value pair in the dataset, as in traditional keyword PIR schemes. 
Database index structures~\cite{anneser2022adaptive,bayer1972symmetric,zhang2022carmi,zhang2024making,DBLP:conf/sigmod/ZhangLAKKP20,DBLP:conf/sigmod/ZhangAPKMS16,ding2020alex,ferragina2020pgm,lan2023updatable,wongkham2022updatable} are exactly designed for such key-to-position conversion. 
In \framework, we propose to apply PGM-indexes to perform this conversion on the client side, allowing direct processing of the querying key in plaintext.
We will discuss the advantages and disadvantages of storing index structures on the client side in~\Cref{sec:client-index}, along with the specific requirements for selecting a suitable index structure for our scenario. Then, we introduce the details of our choice, the \pgm, in~\Cref{sec:pgm}.

\subsection{Client-Side Indexes}\label{sec:client-index}
In \framework, while the index structure used for key-to-position conversion resides on the client side, it is created and maintained by the server, with the client only storing a static copy.
This static copy requires no maintenance on the client side and functions like pre-downloaded hints in other PIR schemes~\cite{celi2024call,davidson2023frodopir,henzinger2023one}.
Sharing this index of the public database poses no security concerns.

The benefits of this approach align with those of PIR schemes, as static index copies can be downloaded during the initialization phase and serve as a local cache to accelerate query execution during the online phase. Specifically, in \framework, the client-side index efficiently maps query keys to predicted locations, which are then expanded using our obfuscated range generation method. This narrows the range of key-value pairs involved in subsequent computations, significantly enhancing query performance. In contrast, prior schemes speed up computations by pre-computing parts of the lookup process but fail to reduce the overall data volume processed.

The main drawback of storing a static copy on the client side is the potential for index staleness due to database updates. Similar to PIR solutions~\cite{patel2023don}, \framework restricts to periodic batch updates for keys (\Cref{sec:update}). Thus the overhead for the client to fetch the latest index from the server is insignificant. 
Moreover, selecting a compact index could mitigate this issue by reducing the time required to download updated versions, further improving efficiency.

Consequently, we have two requirements for the index structure. 
First, to map querying keys to database locations, the index should provide item-level granularity rather than block-level granularity. This allows each data point to be obfuscated at the finest granularity, thus reducing the noise (i.e., the number of fake queries) needed for indistinguishability.
Traditional structures like B+trees~\cite{bayer1972symmetric}, which organize data into pages or large leaf nodes, map querying keys to page IDs. Even if key-value pairs are evenly distributed, noise must be introduced at the page level to satisfy privacy requirements, forcing the obfuscated range to include entire pages and increasing the volume of accessed data.
Second, the size of the index must remain small to reduce the overhead of the client when downloading the index during initialization and after database updates.
Large structures (e.g., hints in the offline phase of many PIR schemes~\cite{celi2024call,henzinger2023one,davidson2023frodopir,zhou2024piano,menon2022spiral}) can become bottlenecks, especially in scenarios where multiple clients simultaneously access the server, leading to bandwidth constraints.

We find that learned indexes can effectively fulfill these two requirements. 
First, learned indexes directly map a querying key to a predicted position in the entire sorted array, providing item-level prediction granularity. We will thoroughly discuss the theoretical differences between adding noise to these two indexes in~\Cref{sec:dp}. 
Second, they consume significantly less memory compared to B+trees, achieving compression rates for non-leaf nodes up to 2,000$\times$ smaller than the internal nodes of a B+tree~\cite{zhang2024making,marcus2020benchmarking,wongkham2022updatable}.

\begin{figure}[!t]
    \centering
    \includegraphics[width=0.95\linewidth]{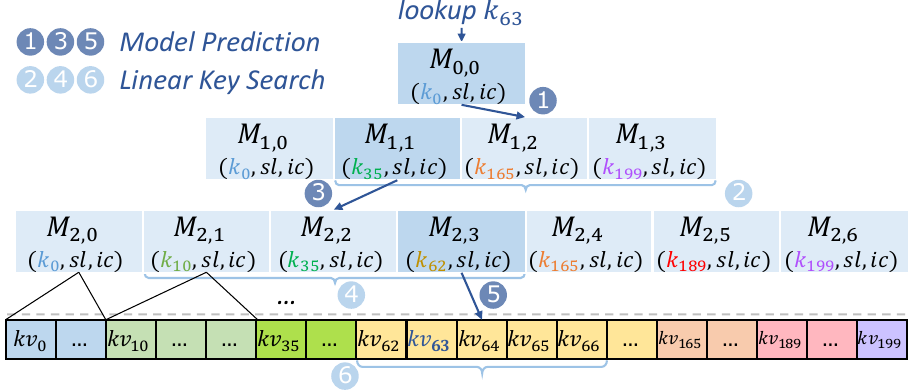}
    \caption{An Example of the \pgm.}
    \label{fig:pgm_index}
\end{figure}

\subsection{PGM-Indexes}\label{sec:pgm}
We integrate a state-of-the-art space-efficient learned index, the \pgm~\cite{ferragina2020pgm}, into \framework to facilitate key-to-position conversion with item-level granularity flexibility. 
As in~\Cref{fig:pgm_index}, the \pgm consists of multiple layers of simple linear regression models, where each model node is defined by only two parameters: slope and intercept (i.e., $\hat{y}=\text{slope}\times k+\text{intercept}$), as well as the minimum key (partitioning key) of the sub-dataset it manages. 
The index has two hyperparameters: $\varepsilon_{\text{model}}$ and $\varepsilon_{\text{data}}$, which determine the maximum allowable prediction errors for all non-leaf nodes and all leaf nodes, respectively.
For a predicted position $\hat{y}$, the desired key-value pair is guaranteed to reside in the range $[\hat{y}-\varepsilon_{\text{data}},\hat{y}+\varepsilon_{\text{data}}]$,
otherwise this pair is not in the key-value store.

We provide an example of the \pgm with $\varepsilon_{\text{model}}=1$ and $\varepsilon_{\text{data}}=2$ in~\Cref{fig:pgm_index}. 
To look up a given key $k_{63}$, model $M_{0,0}$ uses its slope and intercept to predict the next layer's position: $\hat{y} = sl \cdot k_{63} + ic = 2$. The next-layer models within the range $[\hat{y}-\varepsilon_{\text{model}}, \hat{y}+\varepsilon_{\text{model}}] = [M_{1,1},M_{1,3}]$ are searched, locating $M_{1,1}$ (since $k_{35}\leq k_{63}\leq k_{165}$). 
Then, $M_{1,1}$ predicts $[M_{2,1}, M_{2,3}]$ for layer 2 and $M_{2,3}$ is identified similarly. Recursively, $M_{2,3}$ predicts position 64,
narrowing the last-mile search range to $[kv_{62}, kv_{66}]$ with $\pm \varepsilon_{\text{data}}$. Finally, $kv_{63}$ is successfully located using binary search.

For each fixed dataset, the \pgm is constructed using a bottom-up hierarchical approach. The bottom layer scans the ordered key-value pairs with a greedy algorithm of time complexity $O(n)$. 
As described in~\cite{ferragina2020pgm}, constructing the linear models reduces to constructing the convex hull of a set of points. 
Starting at index $i=0$, the algorithm checks whether the current $i$-th key-value pair can be added to the current latest $j$-th model without exceeding $\varepsilon_{\text{data}}$. If true, it moves to the next key-value pair. Otherwise, a linear model (i.e., slope and intercept) is determined by the line that splits the rectangle into two equal-sized halves. Then, a new model $j+1$ is initialized starting from index $i$. 
This ensures that each model manages key-value pairs constrained by a rectangle of height $2\varepsilon_{\text{data}}$. 
This process continues until all key-value pairs are scanned. 
Upper layers follow a similar approach but operate on sub-datasets, where prediction errors are constrained by $\varepsilon_{\text{model}}$. For example, when constructing layer 0, the algorithm processes the pairs $\{\{k_0,0\}, \{k_{35},1\}, \{k_{165},2\}, \{k_{199},3\}\}$. Construction continues until a single model remains at the top layer. More details are provided in~\cite{ferragina2020pgm}.
To maintain simplicity within \framework, we use the default values of the \pgm ($\varepsilon_{\text{data}} = 64$, $\varepsilon_{\text{model}} = 4$) across all datasets, avoiding parameter tuning complexity.

The \pgm is compact, reducing communication overhead during initialization. During the online phase, it can quickly convert a queried key to a predicted range containing the desired key's location.
The size of this range is $2 \times \varepsilon_{\text{data}}+1$, where all positions except the correct one act as fake queries. However, the predicted range does not yet satisfy distance-based indistinguishability, which will be addressed in the next section.
\section{Obfuscated Range Generation}\label{sec:dp}

In this section, we describe how to generate obfuscated ranges that satisfy distance-based indistinguishability. \framework allows users to flexibly specify a relaxed security level for each query. 
It then employs a noise-generation mechanism, which expands the predicted range $[\hat{y}-\varepsilon_{\text{data}},\hat{y}+\varepsilon_{\text{data}}]$, derived from the key-to-position conversion step, to a wider obfuscated range $[l, r]$.
To achieve both high performance and sufficient security, the length of the obfuscated range is set to the minimum value that still ensures distance-based indistinguishability depending on the user-specified security level.
Specifically, any two queries whose keys are within a specified distance $t$ are indistinguishable to the adversary. Larger $t$ requires a wider obfuscated range. 
The obfuscated range $[l, r]$ is then sent to the server, and only data points within this range need to be involved in the computation and communication.



We first employ the exponential mechanism, a widely used approach in differential privacy, to generate noise. 
The exponential mechanism assigns probabilities to points in a given set $\mathbb{D}$, typically based on the distance between the point $i\in \mathbb{D}$ and the real value $x$:
\begin{equation}\label{eq:ldp}
\begin{aligned}
p_{x,i} = \Pr[o = i] = \frac{e^{-\left| x-i \right| \cdot \epsilon_{\text{dp}} / 4t}}{\sum_{j \in \mathbb{D}} e^{-\left| x-j \right| \cdot \epsilon_{\text{dp}} / 4t}}\\
\end{aligned}
\end{equation}
Points closer to the real value would have higher probabilities of being chosen, thereby maintaining indistinguishability while reducing the impact on performance. A point is then randomly sampled as the new boundary based on these probabilities. This process is carried out independently for the left and right boundaries, with $[0, \hat{y}-\varepsilon_{\text{data}}]$ and $[\hat{y}+\varepsilon_{\text{data}}, n)$ as the given $\mathbb{D}$, respectively, and produces the obfuscated range $[l, r]$. We then have the following theorem for the privacy guarantee.

\begin{theorem}\label{theo:ldp}
The exponential mechanism described above provides $\epsilon_{\text{dp}}$-dist indistinguishability privacy guarantee for any pair of values $x, x' \in \mathbb{D}$, where $\left|x-x'\right| \leq t$, and, $t,\epsilon_{\text{dp}} > 0$.
\end{theorem}


\begin{proof}
In our scenario, the observations $O$ visible to the adversary consist of the boundaries (i.e., the outputs $l$ and $r$). We show that our algorithm provides a $\frac{\epsilon_{\text{dp}}}{2}$-dist indistinguishability privacy guarantee for each boundary. 
Let $o$ be one of the two boundaries $l$ or $r$, and $\mathbb{D}$ be the set of potential new values that it is obfuscated to (i.e., $[0, \hat{y}-\varepsilon_{\text{data}}]$ or $[\hat{y}+\varepsilon_{\text{data}}, n)$).
\begin{equation}
\begin{aligned}
\frac{\Pr[O = o | x]}{\Pr[O = o | x']}&= e^{\frac{\epsilon_{\text{dp}}}{4t}(|x'-o|-|x-o|)} \cdot \frac{\sum_{j \in \mathbb{D}} e^{-|x'-j|\cdot\epsilon_{\text{dp}}/4t}}{\sum_{j \in \mathbb{D}} e^{-|x-j|\cdot\epsilon_{\text{dp}}/4t}}\\
\end{aligned}
\end{equation}
By applying the triangle inequality $|x'-o|=|x'-x+x-o|\leq |x'-x|+|x-o|$, we can get:
\begin{equation}
e^{\frac{\epsilon_{\text{dp}}}{4t}(|x'-o|-|x-o|)} \leq e^{\frac{\epsilon_{\text{dp}}}{4t}|x'-x|}\leq e^{\frac{\epsilon_{\text{dp}} \cdot t}{4t}}
\end{equation}
Similarly, $-|x'-j| \leq |x'-x|-|x-j|$, and we get:
\begin{equation}
\begin{aligned}
\sum_{j\in \mathbb{D}} e^{-|x'-j|\cdot\epsilon_{\text{dp}}/4t} &\leq \sum_{j\in \mathbb{D}} (e^{|x'-x|\cdot\epsilon_{\text{dp}}/4t}\cdot e^{-|x-j|\cdot\epsilon_{\text{dp}}/4t})\\
&\leq e^{|x'-x|\cdot\epsilon_{\text{dp}}/4t}\cdot \sum_{j\in \mathbb{D}}e^{-|x-j|\cdot\epsilon_{\text{dp}}/4t}\\
&\leq e^{\frac{\epsilon_{\text{dp}} \cdot t}{4t}}\cdot \sum_{j\in \mathbb{D}}e^{-|x-j|\cdot\epsilon_{\text{dp}}/4t}
\end{aligned}
\end{equation}
Combining both terms, we have the overall bound:
\begin{equation}
\begin{aligned}
\frac{\Pr[O = o | x]}{\Pr[O = o | x']}&\leq e^{\frac{\epsilon_{\text{dp}} \cdot t}{4t}} \cdot e^{\frac{\epsilon_{\text{dp}} \cdot t}{4t}} \leq e^{\frac{\epsilon_{\text{dp}}}{2}}
\end{aligned}
\end{equation}
Since the exponential mechanism is used for both boundaries, the overall algorithm satisfies the $\epsilon_{\text{dp}}$-dist indistinguishability.
\end{proof}

The ratio of the distance $t$ and the privacy parameter $\epsilon_{\text{dp}}$ controls the size of the obfuscated range. For smaller $\frac{t}{\epsilon_{\text{dp}}}$, the obfuscated range is probabilistically reduced, which provides more relaxed privacy but saves the computational and communication costs.

It is worth noting that the exponential mechanism involves calculating probabilities for all points in $\mathbb{D}$ and then sampling from them, which is costly. 
To address this issue, we simulate the exponential mechanism by employing the widely used discrete Laplace mechanism~\cite{li2022opboost}, as defined in~\Cref{def:lap}, with $\lambda = 2t/\epsilon_{\text{dp}}$.
The process for obfuscated range generation using the discrete Laplace mechanism is outlined in~\Cref{alg:lap}.

\begin{algorithm}[!t]
\caption{Lap$_{\mathbb{Z}}$-based Obfuscated Range Generation}\label{alg:lap}
\begin{algorithmic}[1]
\Require $\hat{y}$, $\varepsilon_{\text{data}}$, dataset size $n$, privacy parameter $\epsilon_{\text{dp}} > 0$, $t>0$
\Ensure obfuscated range $[l, r]$
\Function{ModularLap$_{\mathbb{Z}}$}{$ \mathbb{D}, \epsilon_{\text{dp}}, t$}
    \State $x \leftarrow Lap_\mathbb{Z}(\frac{2t}{\epsilon_{\text{dp}}})$
    \State $x\leftarrow x\% \left|\mathbb{D}\right|$
    \State \Return $\mathbb{D}[x]$
\EndFunction
\State $\mathbb{D}_l\leftarrow [\hat{y}-\varepsilon_{\text{data}}, \cdots, 1, 0, n-1,n-2, \cdots, \hat{y}+\varepsilon_{\text{data}}]$ 
\State $\mathbb{D}_r\leftarrow [\hat{y}+\varepsilon_{\text{data}}, \cdots, n-2, n-1, 0,1,\cdots, \hat{y}-\varepsilon_{\text{data}}]$
\State $l\leftarrow \text{ModularLap}_{\mathbb{Z}}(\mathbb{D}_l, \epsilon_{\text{dp}}, t)$
\State $r\leftarrow \text{ModularLap}_{\mathbb{Z}}(\mathbb{D}_r, \epsilon_{\text{dp}}, t)$
\If{$l\leq \hat{y}-\varepsilon_{\text{data}}\leq \hat{y}+\varepsilon_{\text{data}}\leq r$}
    \State \Return $[l, r]$
\ElsIf{$r<l\leq \hat{y}-\varepsilon_{\text{data}}$ \textbf{or} $\hat{y}+\varepsilon_{\text{data}}\leq r<l$}
    \State \Return $[l,n-1] \cup [0, r]$
\Else
    \State \Return $[0, n-1]$
\EndIf
\end{algorithmic}
\end{algorithm}

\begin{definition}[Discrete Laplace Distribution]\label{def:lap}
The discrete Laplace distribution with a scale parameter $\lambda$ is denoted as $\mathrm{Lap}_\mathbb{Z}(\lambda)$, 
where $\mathbb{Z}$ represent the set of integers. If a random variable $X$ follows $\text{Lap}_\mathbb{Z}(\lambda)$, its probability distribution is defined as: 
\begin{equation}
\forall x\in \mathbb{Z},\quad \textnormal{Pr}[X=x]=\frac{e^{1/\lambda}-1}{e^{1/\lambda}+1}\cdot e^{-\left|x\right| / \lambda}
\end{equation}
\end{definition}

We then prove that adding noise according to a discrete Laplace distribution satisfies distance-based indistinguishability.
\begin{theorem}\label{theo:lap_dp}
\textbf{Lap$_{\mathbb{Z}}$-based Obfuscated Range Generation} provides $\epsilon_{\text{dp}}$-dist indistinguishability privacy guarantee for any pair of values $x_1, x_2$, where $\left|x_1-x_2\right| \leq t$, and $t,\epsilon_{\text{dp}} > 0$.
\end{theorem}
\begin{proof}
This algorithm provides a $\frac{\epsilon_{\text{dp}}}{2}$-dist indistinguishability privacy guarantee for each boundary (i.e., $\hat{y}-\varepsilon_{\text{data}}$, $\hat{y}+\varepsilon_{\text{data}}$). 
Let $l_1$ and $l_2$ be the left boundaries of $x_1$ and $x_2$. 
The probability ratio of $l_1$ and $l_2$ being randomized to the same output value $o$ is: 
\begin{equation}
\begin{aligned}
\frac{\Pr[l_1+N_1]}{\Pr[l_2+N_2]}&=\frac{\Pr[N_1=o-l_1]}{\Pr[N_2=o-l_2]} = \frac{e^{-\left|o-l_1\right|/\lambda}}{e^{-\left|o-l_2\right|/\lambda}} \leq e^{\frac{\epsilon_{\text{dp}}}{2}} \\
\end{aligned}
\end{equation}
where $N_1$ and $N_2$ follow $\text{Lap}_\mathbb{Z}(\lambda)$.
Since this mechanism is used for both boundaries, the overall algorithm satisfies the $\epsilon_{\text{dp}}$-dist indistinguishability: $\epsilon_{\text{overall}} = \epsilon_l + \epsilon_r=\epsilon_{\text{dp}}/2 + \epsilon_{\text{dp}}/2=\epsilon_{\text{dp}}$.
\footnote{Note that our proof now demonstrates that distance-based indistinguishability holds for two boundaries derived from the predicted positions ($\hat{y}$). For practical use, the user should set $t' = t - 2\varepsilon_{\text{data}}$ as the final distance. Since $\left|\hat{y}_i - \hat{y}_j\right| \leq \left|i - j\right| + 2\varepsilon_{\text{data}} \leq t + 2\varepsilon_{\text{data}}$, this guarantees that the proposed mechanism preserves distance-based indistinguishability when applied to queries.}
\end{proof}

\Cref{theo:lap_dp} guarantees that for any two queries, $q_i$ and $q_j$, with distance less than $t$, the probability of distinguishing between them is bounded by $e^{\epsilon_{\text{dp}}}$. 
This covers two cases: when $q_i \neq q_j$ and when the same query is submitted repeatedly (i.e., $q_i = q_j$). In both cases, the adversary will be unable to distinguish the queries.

Since the negative numbers sampled from $\mathrm{Lap}_\mathbb{Z}(\lambda)$ are modulo positive, the expected amount of noise added to one boundary is $\mathbb{E}[X \mid X > 0] = \frac{1}{1 - e^{-1/\lambda}} \approx 2t/\epsilon_\text{dp}$. 
Accounting for both boundaries, the expected length of the obfuscated range becomes $4t/\epsilon_\text{dp} + 2\varepsilon_{\text{data}} + 1$, where the last two terms specify the output range from the \pgm. This length is capped at the total number of key-value pairs, i.e., $n$. The total number of fake queries equals this length minus one.

\paragraph{Discussion}
Below, we analyze why B+tree requires more noise compared to PGM-indexes. 
The B+tree index only needs to transfer non-leaf nodes to the client, who uses it to obtain the ID of the leaf node (i.e., page ID) containing the querying key, as well as the IDs of the key-value pairs at the left and right boundaries of this page ($l_{\text{B+tree}}$ and $r_{\text{B+tree}}$). 
Note that noise must be added at the page-level granularity. Otherwise, the adversary may infer that the querying key is not located in any incomplete page within the range (e.g., the first or last page), leaking information.

Assume $m$ is the number of key-value pairs per page, and key-value pairs are evenly distributed within each page. Thus $\text{page ID} =\frac{\text{item ID}}{m}$.
The distance in the B+tree can be defined as $t'=\lceil \frac{t}{m} \rceil$. The boundaries $l_{\text{B+tree}}$ and $r_{\text{B+tree}}$ are then extended using the same noise mechanism ($\mathrm{Lap}_\mathbb{Z}(\lambda')$), where $\lambda' = 2t' / \epsilon_\text{dp}$. Consequently, the expected length of the obfuscated range becomes $\frac{4m}{\epsilon_\text{dp}} \cdot \lceil \frac{t}{m} \rceil + m$.

Compared to the \pgm, the B+tree generally requires more noise. For example, the leaf nodes of a B+tree typically occupy 4~KB, containing 256 key-value pairs of 16 bytes each. When using default parameters ($\varepsilon_{\text{data}} = 64$, $\epsilon_{\text{dp}} = 2^{-6}$) and a user-selected $t = 100$, the expected obfuscation range length for a B+tree is 65,792, while for a \pgm, it is 25,729. When $t = 10,000$, the expected lengths increase to 2,621,696 for a B+tree and 2,560,129 for a \pgm, respectively, indicating that the \pgm requires less noise overall, especially for fine-grained security levels. More empirical results are provided in~\Cref{sec:experiment_pgm}.
\section{Key-Value Retrieval}\label{sec:retrieval}
In this section, we describe our server-side adaptive key-value retrieval module, including two schemes: plaintext download and variable-range PIR. Additionally, we present a lightweight cost model to help make fast and efficient decisions between these two schemes based on query characteristics and system configurations.

\subsection{Plaintext Download}\label{sec:download}
Plaintext download is a straightforward method for retrieving key-value pairs, where the server directly sends the key-value pairs within the obfuscated range to the client, which performs the searches locally. The client only needs to provide the boundaries $l$ and $r$. 
This approach satisfies the security requirements, as the server only sees the boundaries that have been obfuscated to ensure distance-based indistinguishability, and remains unaware of the exact key the client intends to retrieve. 

The main advantage of plaintext download is server-side efficiency, as the server can immediately transmit the specified key-value pairs over the network without additional computation. 
However, this approach is bandwidth-intensive, which can lead to performance bottlenecks in low-bandwidth or high-traffic environments.
Besides, plaintext download can further reduce the data to be transferred through compression techniques~\cite{qiao2024blitz,liu2024leco,AbadiMF2006integrating,DBLP:journals/pvldb/ZengHSPMZ23}. Although we do not apply compression in our experiments, it can be integrated as an optional optimization.

\subsection{Variable-range PIR}\label{sec:variablePIR}
To handle low-bandwidth scenarios, we propose a novel variable-range PIR scheme. 
This enhanced scheme allows performing ``tiny PIR'' on a specific range of each query without re-preprocessing. 
Our design significantly improves server-side computational efficiency when querying large datasets with relaxed privacy levels.

\subsubsection{Underlying Cryptography.}
We leverage the SEAL FHE library, which is based on the BFV FHE cryptosystem, to perform ``tiny PIR''.
This method encodes several data points into a single plaintext, enabling vectorized homomorphic operations. Furthermore, each retrieval returns the key-value pairs packed within the same plaintext, facilitating simultaneous retrieval of multiple key-value pairs.
The number of key-value pairs (denoted as $M$) that a single plaintext can hold is determined by: $M =\frac{N}{\lceil kv_{\text{bits}} / \log_2 p \rceil}$, where the polynomial degree $N$ and the plaintext modulus $p$ are both internal parameters of the FHE scheme. The modulus $q$, which governs the noise capacity and the overall security, can be derived from $N$ and $p$. We use the recommended default values of $N = 4096$ and $\log_2 p = 20$. More details can be found in the original paper~\cite{fan2012somewhat,chen2017simple,angel2018pir}.

\begin{figure}[!t]
    \centering
    \includegraphics[width=1\linewidth]{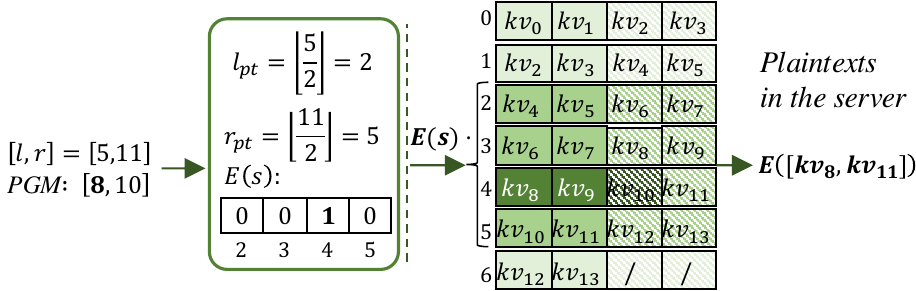}
    \caption{An Example of Data Encoding and Query Processing in Variable-Range PIR.}
    \label{fig:variablePIR}
\end{figure}

\subsubsection{PGM-Oriented Misaligned Encoding}
One specific challenge in our framework is that we must return the key-value pairs in the range $[\hat{y}-\varepsilon_{\text{data}},\hat{y}+\varepsilon_{\text{data}}]$ as predicted by the \pgm, rather than a single record as in classical PIR schemes. 
Without a careful design, the key-value pairs in this range may cross two or multiple plaintexts, all of which need to be transmitted to the client. 
For example, if $M=4$ and we need to fetch $[kv_6, kv_9]$, two plaintexts corresponding to $[kv_4, kv_7]$ and $[kv_8, kv_{11}]$ must be returned. 

We introduce a PGM-oriented misaligned encoding scheme. As illustrated in~\Cref{fig:variablePIR}, only $M/2$ key-value pairs are encoded into each plaintext, with the second half containing a duplicate of the first $M/2$ pairs from the next plaintext. 
This ensures that, as long as the \pgm satisfies $2\varepsilon_{\text{data}}+1 < \frac{M}{2}+2$, 
all required key-value pairs can be retrieved in a single plaintext.
This condition is easily satisfied with typical configurations, e.g., usually $\varepsilon_{\text{data}} = 64$ and $M = \frac{4096}{\lceil 128/20 \rceil}=585$.
Furthermore, this misaligned encoding can be generalized to support longer predicted ranges by increasing the overlap between adjacent plaintexts. With $m$ key-value pairs overlapped, the condition would become $ 2\varepsilon_{\text{data}} + 1 < m + 2 $.

\subsubsection{Query Processing}
Recall that $[\hat{y}-\varepsilon_{\text{data}}, \hat{y}+\varepsilon_{\text{data}}]$ is enlarged to $[l,r]$ to ensure distance-based indistinguishability. While only the key-value pairs in the first (smaller) range are returned to the user, all the pairs in the second (larger) range must be uniformly accessed and processed to obscure sensitive access patterns. 
The size of this larger range can vary significantly with different user-specified security requirements. 
This requires that the underlying PIR scheme should encode plaintexts independently (e.g., SealPIR~\cite{angel2018pir} and OnionPIR~\cite{mughees2021onionpir}) without pre-computation involving the client. In this paper, we use SealPIR. But \framework can incorporate more advanced PIR interfaces in the future.

\Cref{fig:variablePIR} illustrates the steps in processing a lookup query with our variable-range PIR based on the misaligned encoding. 
Both the predicted and obfuscated ranges are first converted to the plaintext IDs at the client side using the encoding information (e.g., $M$) obtained during initialization. 
Here with $M=4$, $l=5$, and $r=11$, the plaintext IDs are $l_{\text{pt}}=2$ and $r_{\text{pt}}=5$, involving four plaintexts in the computation. 
The predicted range is encoded into a one-hot vector $s$ of length $(r_{\text{pt}} - l_{\text{pt}} + 1)$, where 1 indicates the target plaintext and 0 represents the others. 
In this case, the third element of $s$ is set to 1, corresponding to the 4th plaintext. 
The client sends $l_{\text{pt}}$, $r_{\text{pt}}$, and the homomorphic encryption of $s$, $E(s)$, to the server. 
Note that the server performs homomorphic operations only on the obfuscated range. Specifically,
a homomorphic inner product is performed between $E(s)$ and the plaintexts in $[l_{\text{pt}}, r_{\text{pt}}]$, producing an encrypted output containing the retrieved key-value pairs (i.e., $[kv_8, kv_{11}]$), which is then returned to the client.

Our variable-range PIR enables relaxed lookup queries with flexible security requirements without requiring re-encoding, improving utility and scalability. Additionally, misaligned encoding ensures efficient retrieval of key-value pairs in a single PIR call, reducing communication overheads. 
However, compared to plaintext downloads, it increases server-side computation due to the homomorphic encryption operations required for PIR requests.

\subsection{Scheme Selection}
We employ a cost model to efficiently choose between the two schemes. The model calculates the latency for each scheme using the following formula, as latency is the primary concern for clients: 
\begin{equation}
C = \frac{\textit{Comm}}{\text{bandwidth}} + C_{\text{compute}}
\end{equation}
where $\textit{Comm}$ represents the total communication amount of the scheme, and $C_{\text{compute}}$ denotes the server-side computation time. This cost model is broadly applicable to many retrieval schemes.

For plaintext downloads, the communication cost is $w \times kv_{\text{bits}}$, where $w$ is the length of the obfuscated range, and the computation time is negligible. In contrast, for variable-range PIR, the communication cost includes the encrypted one-hot vector and the ciphertext returned by the server. Its computation time is $w\times C_{\text{FHE}}$, where $ C_{\text{FHE}} $ represents the time needed to perform a homomorphic computation operation on a single plaintext, which increases with the number of plaintexts involved.

Both bandwidth and $ C_{\text{FHE}} $ can be pre-determined, facilitating quick decisions during the online phase. By leveraging this cost model, we can efficiently combine different retrieval schemes, enabling faster privacy-preserving retrieval for users.

\section{Supporting Updates}\label{sec:update}
As discussed in \Cref{sec:motiv}, in our scenarios, database updates are performed by the server, not by the clients. This setting aligns with existing PIR-based designs~\cite{davidson2023frodopir,henzinger2023one,celi2024call,patel2023don}, in which any modification necessitates re-initializing the entire PIR scheme, including re-running the offline phase. In \framework, we support two strategies to handle updates, i.e., real-time updates for \textit{values} and periodic batch updates for \textit{keys}. We also leverage multi-version control to synchronize server updates with clients.

For value updates, the plaintext of the updated key-value pair is retrieved from the key-value store, modified in real-time, and then stored back. Since the \pgm is constructed based on keys only, it remains unchanged during value updates. Throughout this process, client lookup queries continue to operate as usual. If the obfuscated range of a lookup query overlaps with the updated plaintext, read-write locks are employed to ensure consistency.

For key updates, we employ periodic batch updates via a copy-on-write model. In this approach, the server initiates a new cloud instance or background thread to merge the existing key-value store with new key-value pairs, performing necessary inserts, deletes, and re-encoding operations. Only key-value pairs that need to be moved to maintain order are re-encoded. For example, if an insertion occurs at position 7 in a dataset of 10 key-value pairs, the plaintexts of the first 6 pairs remain unchanged and do not require re-encoding. Once the merge is complete, the server constructs a new \pgm for the updated dataset. This process is done asynchronously (or offline) and does not disturb services on older versions. After the update, the server broadcasts the new \pgm and version information to all connected clients, or clients can request it on demand. If a client queries using the old version, the server returns results from that version while notifying the client of the version inconsistency and providing the latest \pgm. The client can then either use the old version's results or re-execute the query using the new index. In our experiments, clients re-execute queries with the latest version. Once all clients have confirmed that they have switched to the new version, the old key-value store is deleted.
\section{Evaluation}\label{sec:experiments}
In this section, we conduct a comprehensive evaluation of \framework using experiments on multiple datasets. \Cref{sec:setup} outlines the experimental setup, while~\Cref{sec:end2end} provides an in-depth analysis of end-to-end performance. \Cref{sec:impact_params} presents the explorations of the impact of various factors. Finally, \Cref{sec:eva_update} evaluates a realistic scenario involving Redis and update operations.


\subsection{Experimental Setup}\label{sec:setup}

We conduct experiments on a machine with two Intel$^\circledR$ Xeon$^\circledR$ Platinum 8474C CPUs (2.10 GHz, 48 cores per socket) and 512 GB of RAM. 
All end-to-end evaluations in~\Cref{sec:end2end} are performed using 8 threads, and the remaining experiments use a single thread. 
We select the commonly used WAN setting (50 Mbps bandwidth with 30 ms roundtrip latency) as the default configuration~\cite{mohassel2018aby3,keller2016mascot,keller2018overdrive}. We provide two additional bandwidth setups in~\Cref{sec:bandwidth}, 100 Mbps and 10 Mbps, all with a roundtrip latency of 30 ms.

\para{Datasets and Workloads.}
We use the real-world Open Street Map Coordinates (OSMC) dataset from the SOSD benchmark~\cite{marcus2020benchmarking}. It contains 200 million key-value pairs representing real geographic positions. In contrast, existing PIRs~\cite{ahmad2022pantheon,celi2024call,patel2023don} often use randomly generated datasets with at most a few million data points and thus fail to fully capture the complexity of large-scale, real-world applications. Both keys and values are 8-byte unsigned integers. This uniform length, consistent with PIR schemes, prevents the server from inferring query keys based on lengths. We use the entire 200M key-value pairs in \Cref{sec:end2end}. We randomly select 100 keys (similar to \cite{celi2024call}) from the dataset for client lookups during the online phase. We use a subset of the datasets in other experiments.

\para{Metrics.}
Our evaluation includes four main metrics: offline/online communication volume and offline/online execution time. Communication volume refers to the total amount of data transferred between the server and the client at each phase, measured in megabytes (MB). The execution time includes both the client/server computation time and the network data transmission time, providing an end-to-end performance measure.

\para{\framework and Baselines.}
We implement \framework and two keyword-PIR baselines, Chalamet~\cite{celi2024call} and Pantheon~\cite{ahmad2022pantheon}, in Java. 
Since the original Rust implementation of Chalamet does not include network communication, it cannot be evaluated directly under our client-server experiment setting. To enable a fair end-to-end evaluation, we reimplemented Chalamet in Java\footnote{In our tests with a $2^{20} \times 256$ byte dataset, we observe that network communication, along with serialization and deserialization, account for more than 72\% of the total execution time. Since the computational functionality of the Rust code is not the bottleneck, rewriting it in Java has minimal impact on the overall performance.} based on its open-source implementation~\cite{chalamet}.
Both Chalamet and Pantheon are fully secure schemes, while \framework can support different security granularities. 
As discussed in~\Cref{sec:privacy},
we keep $\epsilon_{\text{dp}} = 2^{-6}$, and specify three representative security levels as high ($t=1,000,000$), medium ($t=10,000$), and low ($t=100$). 
We build the PGM-indexes using the complete datasets of the respective size (e.g., 200 million key-value pairs in the end-to-end evaluation) with parameters set to default values ($\varepsilon_{\text{model}} = 4, \varepsilon_{\text{data}} = 64$).
To support variable-range PIR in \framework, we employ Java Native Interface (JNI) technology to call the Microsoft SEAL library (v4.1)~\cite{sealcrypto} and use default settings for the BFV scheme. 
Our source code is publicly available~\cite{femur}.

\subsection{End-to-End Evaluation}\label{sec:end2end}
We perform an end-to-end evaluation using the 200M-record OSMC dataset. In this experiment, the client continuously issues 100 queries in a pipeline manner to the server. We compare the total execution time (from the generation of the first query to the client receiving all responses) of \framework against the baselines. \framework is evaluated at full security and eight different security levels with varying $t$ values of 1M (high), 0.5M, 0.1M, 50K, 10K (medium), 5K, 1K, and 100 (low), which correspond to the average number of data points in \framework's computations of 256M\footnote{Note that this is a theoretical expectation. Since the OSMC dataset contains only 200M keys, any obfuscated range exceeding 200M will be capped at 200M.}, 128M, 25.6M, 12.8M, 2.56M, 1.28M, 256K, and 25.6K data points, respectively (plus an additional 129 data points generated by the \pgm for each case).

\begin{figure}[!t]
    \centering
    \includegraphics[width=1\linewidth]{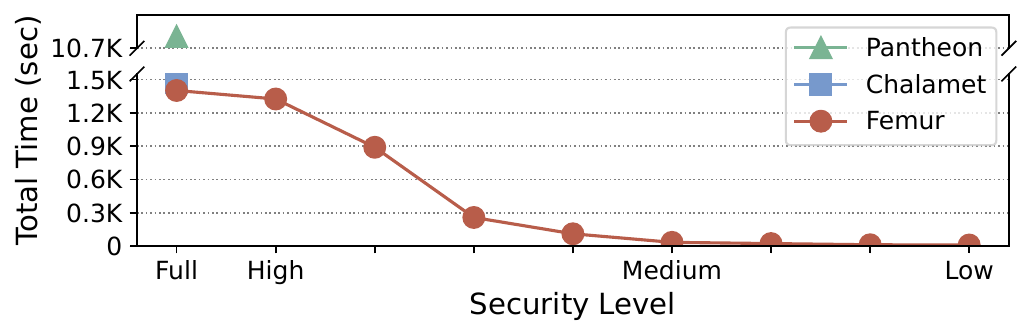}
    \caption{Total Online Execution Time for 100 Queries\textnormal{ -- \framework uses 8 relaxed security levels besides full security. Pantheon and Chalamet encountered out-of-memory issues on 8 threads, so their results are ideally scaled from single-thread time. Pantheon failed to complete on this large dataset after running 24 hours using a single thread, so we mark its (lower-bound) time as 24/8 = 3 hours.}}
    \label{fig:end2end}
\end{figure}

We present the online execution time for \framework and the two baselines in~\Cref{fig:end2end}, and their offline execution time in \Cref{tab:end2end_offline}. The results clearly show that relaxing the security significantly enhances the online runtime, with speedups of 1.11$\times$, 1.65$\times$, 5.71$\times$, 13.4$\times$, 44.6$\times$, 66.8$\times$, 131.1$\times$, and 163.9$\times$ over Chalamet, which provides full security. 
Even at the same full security level, \framework is slightly faster than Chalamet (1.05$\times$) and significantly outperforms Pantheon (over 7.71$\times$). 
\framework reduces the average execution time per query from 14 seconds under full security to 89.6 ms at a relaxed security level that guarantees indistinguishability among 100 neighboring queries. This demonstrates an effective trade-off between security and performance, enabling significant speedups when users are willing to relax their privacy requirements.

\begin{table}[!t]
    \centering
    \setlength{\tabcolsep}{1.2pt} 
    \renewcommand{\arraystretch}{1} 
    \small 
    \caption{Total Offline Time for 100 Queries (in seconds)\textnormal{ --  \framework only needs to be initialized once to support all security-level queries, so different $t$ values produce similar offline latencies.}}
    \label{tab:end2end_offline}
    \begin{tabular}{c c c c c c c c c c}
    \toprule
    \multirow{2}{*}[-0.5ex]{\textbf{Pantheon}}&\multirow{2}{*}[-0.5ex]{\textbf{Chalamet}}&\multicolumn{8}{c}{\textbf{\framework}} \\
    &&\multicolumn{8}{c}{With Relaxed Security Level (Various $t$)}\\
    \cmidrule(lr){1-2} \cmidrule(lr){3-10}
    \multicolumn{2}{c}{\footnotesize Full Security} &\footnotesize 1M &\footnotesize 500K &\footnotesize 100K &\footnotesize 50K &\footnotesize 10K  &\footnotesize 5K &\footnotesize 1K &\footnotesize 100 \\
    \midrule
    55.5& 16214.3& 697.5&697.2&703.2&697.3&698.2&694.6&702.1&696.2 \\
    \bottomrule
    \end{tabular}
\end{table}

Pantheon is extremely slow on large datasets, primarily due to its reliance on a slow homomorphic equality check during the online phase to locate the querying key. In contrast, \framework employs the \pgm for rapid key location, at the slight cost of sending a small amount of data during initialization. Even with the default value of $\varepsilon_{\text{data}} = 64$, which is relatively small, the size of the \pgm is only 5MB on the 200 million records.
While Chalamet performs close to \framework with full security level, it must transfer 858.5 MB of ciphertext results, which can significantly delay query responses in bandwidth-constrained scenarios. Unlike Pantheon and Chalamet, whose computation or transmission times grow with dataset size, \framework maintains a constant online transfer time, limited to two ciphertexts (for query and response), with computation time dependent solely on the user-specified security level.

Chalamet suffers from a significantly slower offline phase, requiring 4.5 hours to pre-process 200 million key-value pairs as shown in~\Cref{tab:end2end_offline}. This makes Chalamet impractical in scenarios where updates or inserts are needed, as even batch updates require re-processing all data points, leading to substantial delays. 
In contrast, \framework completes the offline phase in 11.5 minutes, significantly faster than Chalamet. 
Besides, our update strategies in~\Cref{sec:update} further reduce the re-initialization time by minimizing the number of key-value pairs that need to be re-encoded. Further evaluation of update support is presented in~\Cref{sec:eva_update}.
Note that to support the flexibility of different privacy levels in \framework, we do \emph{not} need to re-execute the initialization phase. 
Pantheon’s offline initialization time is short because it only requires the user to upload the two parameters for the underlying homomorphic encryption scheme. However, its online query processing time is significantly longer.

\subsection{Impact of Various Factors}\label{sec:impact_params}
Starting from this section, we focus on the performance breakdown of individual queries and the impact of various factors on performance.
Therefore, each query is run individually with a single thread, and the online execution time reflects the performance of a single query. Unless otherwise stated, other parameters remain the same as previously described.

\subsubsection{Impact of Dataset Sizes}\label{sec:datasize}
\begin{figure}[!t]
    \centering
    \includegraphics[width=1\linewidth]{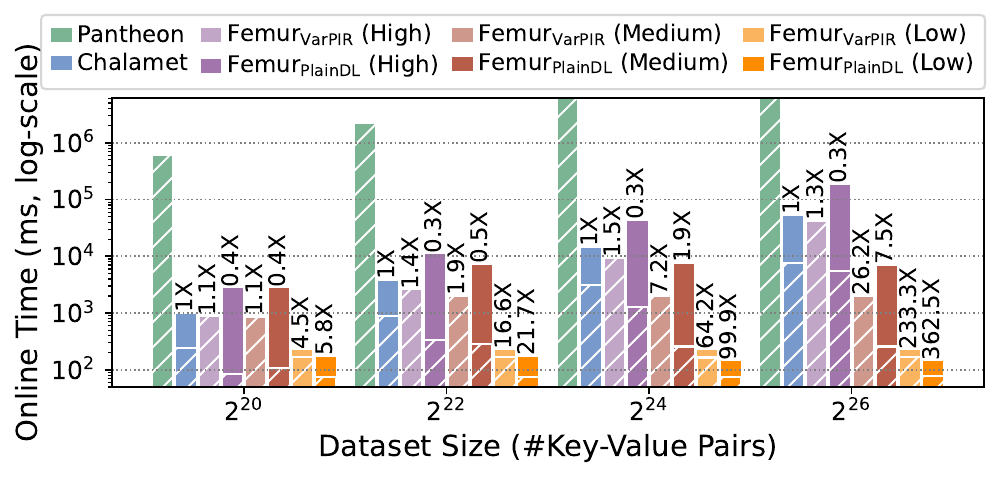}
    \caption{Online Execution Time for Each Query on Different Dataset Sizes\textnormal{ -- The number above each bar represents the speedup over Chalamet. Pantheon timed out on datasets of size $2^{24}$ and $2^{26}$. The shadowed area represents the server-side computation time, including the time required to serialize the data to be sent.}}
    \label{fig:size_online_time}
\end{figure}

In this section, we evaluate the performance of each scheme across various dataset sizes. Specifically, we conduct experiments on a subset of the OSMC dataset of sizes $2^{20}$, $2^{22}$, $2^{24}$, and $2^{26}$. 
We also test smaller datasets, which are not included here due to space limitations, and observed conclusions consistent with the results presented.
We select three security levels for \framework, corresponding to the expected length of obfuscated ranges of 256 million, 2.56 million, and 25600, to evaluate both \framework$_\text{PlainDL}$ and \framework$_\text{VarPIR}$ schemes. If the expected length exceeds the total number of data points in the dataset, the current scheme is the same as full security (i.e., high level for all, and medium level for $2^{20}$).

\begin{table}[!t]
    \centering
    \setlength{\tabcolsep}{1.2pt} 
    \renewcommand{\arraystretch}{1.1} 
    \small 
    \caption{Online Communication per Query (in MB)\textnormal{ -- including the query keys/parameters uploaded by the client and the query results downloaded from the server.}}
    \label{tab:size_online_comm}
    \begin{tabular}{c c c c c c c c c}
    \toprule
    \multirow{2}{*}[-0.5ex]{\textbf{\shortstack{Dataset\\Size}}} & \multicolumn{2}{c}{\shortstack{Baselines on\\Full Security}} & \multicolumn{2}{c}{\shortstack{\textbf{\framework} on\\High Security}} & \multicolumn{2}{c}{\shortstack{\textbf{\framework} on\\Medium Security}} & \multicolumn{2}{c}{\shortstack{\textbf{\framework} on\\Low Security}} \\
    \cmidrule(lr){2-3} \cmidrule(lr){4-5} \cmidrule(lr){6-7} \cmidrule(lr){8-9}
    & \scriptsize \textbf{Pantheon} & \scriptsize \textbf{Chalamet} & \scriptsize \textbf{VarPIR} & \scriptsize \textbf{PlainDL} & \scriptsize \textbf{VarPIR} & \scriptsize \textbf{PlainDL} & \scriptsize \textbf{VarPIR} & \scriptsize \textbf{PlainDL} \\
    \midrule
    $2^{20}$  & 3.5  & 4.5   & 0.43  & 16    & 0.43  & 16   & 0.43  & 0.39 \\
    $2^{22}$  & 3.5  & 18    & 0.43  & 64    & 0.43  & 39   & 0.43  & 0.39 \\
    $2^{24}$ & 3.5  & 72    & 0.43  & 256   & 0.43  & 39   & 0.43  & 0.39 \\
    $2^{26}$ & 3.5  & 288   & 0.43  & 1024  & 0.43  & 39   & 0.43  & 0.39 \\
    \bottomrule
    \end{tabular}
\end{table}
\Cref{fig:size_online_time} illustrates the online time for each method across various dataset sizes, 
which mainly consists of the server-side computation time and the communication time to transmit data between the two parties.
The online latency of Chalamet increases linearly with dataset size, rising from 1.02 seconds for $2^{20}$ key-value pairs to 53.7 seconds for $2^{26}$ key-value pairs. This linear growth is not only due to the increased computation but also due to the linear increase in online communication, as shown in ~\Cref{tab:size_online_comm}. Pantheon shows a significant computation time despite constant communication volumes (3.5 MB), with one query on $2^{22}$ entries taking an impractical 2,261 seconds. 
In contrast, \framework's online response time remains independent of dataset size at fixed security granularity. This is because the security requirements for the client typically do not escalate with the dataset size. 
Compared to Chalamet, \framework achieves a maximum speedup of 362.5$\times$ by relaxing security to ensure distance-based indistinguishability with $t=100$ and $\epsilon_{\text{dp}}=2^{-6}$. 
Furthermore, the online communication cost of \framework$_\text{VarPIR}$ is fixed at 0.43 MB (i.e., the size of a ciphertext), regardless of dataset size or security granularity. Similarly, the online communication cost of \framework$_\text{PlainDL}$ does not increase with dataset size for the same security levels, remaining at 0.39 MB for all low-security experiments. Note that for higher security levels in~\Cref{tab:size_online_comm}, the expected length exceeds the size of certain datasets, causing the online communication to degrade to the entire dataset size.

\begin{figure}[!t]
    \centering
    \includegraphics[width=1\linewidth]{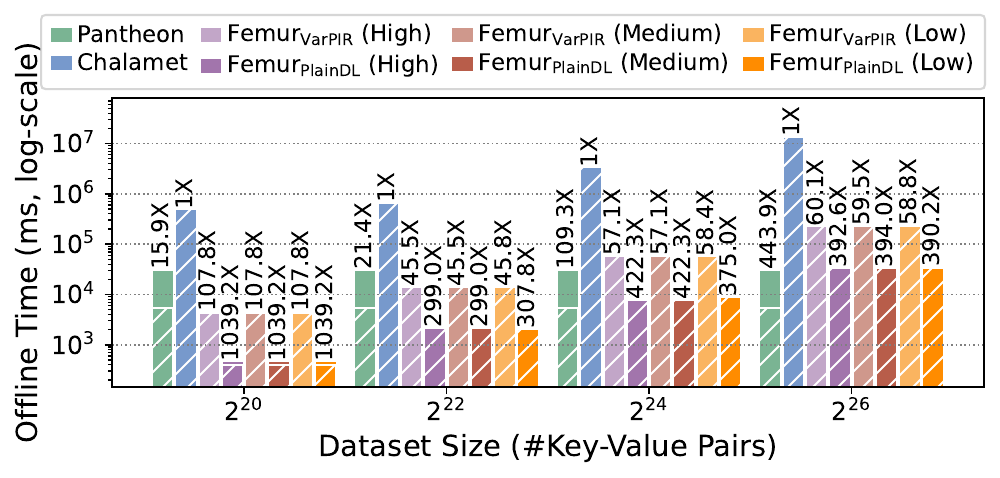}
    \caption{Offline Time for Each Query on Different Dataset Sizes\textnormal{ -- The number above each bar is the speedup ratio of each scheme compared to Chalamet.}}
    \label{fig:size_offline_time}
\end{figure}

Among our two schemes, \framework$_\text{PlainDL}$ scheme is more sensitive to bandwidth. It does not involve plaintext/ciphertext computations, so its query time is determined by the bandwidth and the byte size of the data. 
In contrast, \framework$_\text{VarPIR}$ is compute-intensive. Owing to our encoding design, the network transmission involves only uploading a query and downloading a result, whose time is minimized and remains constant at 63 ms. Consequently, its overall runtime is almost dominated by computation. 
In \framework, the cost model automatically selects the faster scheme between the two.

Particularly, Chalamet and \framework$_\text{VarPIR}$ (both High and Medium) on the dataset size with $2^{20}$ provide comparable online end-to-end latency, with the same level of full security.
However, in the offline phase, our \framework$_\text{VarPIR}$ offers a 107.8$\times$ speedup compared to Chalamet, as shown in~\Cref{fig:size_offline_time}, by only encoding the records into FHE plaintexts and efficiently building and transmitting the \pgm. Since the shorter offline phase of \framework$_\text{PlainDL}$ only accomplishes tasks that are a subset of \framework$_\text{VarPIR}$ (i.e., transmit \pgm), \framework's overall offline time is determined by \framework$_\text{VarPIR}$. Notably, our offline time remains constant across different security levels, as it depends solely on the size of the public key-value store. Clients download the \pgm once and can then use \framework for queries with different security requirements, whereas Chalamet and Pantheon rigidly support only fully secure queries.

\subsubsection{Impact of Bandwidth}\label{sec:bandwidth}
\Cref{fig:bw} illustrates the online latency across different bandwidths for the dataset of $2^{26}$ key-value pairs. We evaluate three security levels of \framework: high, medium, and low.
The high level aligns with the full security level of the baselines, as the expected length of the obfuscated range exceeds $2^{26}$. 
Bandwidth has less impact on \framework$_\text{VarPIR}$ compared to both \framework$_\text{PlainDL}$ and Chalamet. 
\framework$_\text{PlainDL}$ is bandwidth-intensive and becomes faster with higher bandwidth. 
Chalamet, which transfers 288 MB of data, is similarly limited by bandwidth. At 10 Mbps, it takes 250.5 seconds to respond to a query. 
In contrast, \framework$_\text{VarPIR}$ is able to transfer only one ciphertext due to its misaligned encoding, and it is bound by the FHE computations rather than data transfers. 

\begin{figure}[!t]
    \centering
    \includegraphics[width=1\linewidth]{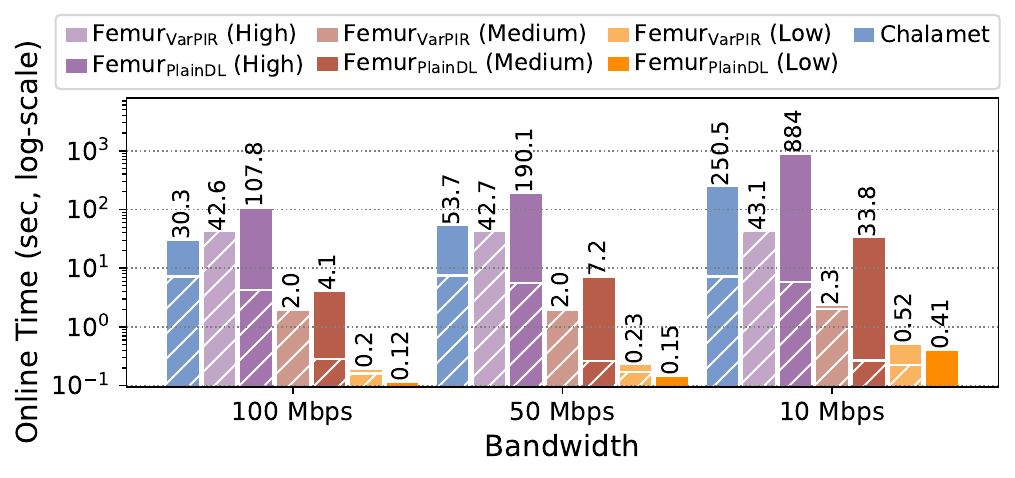}
    \caption{Online Execution Time per Query at Various Bandwidths\textnormal{ -- The number above each bar is the latency in seconds.}}
    \label{fig:bw}
\end{figure}

In summary, \framework, leveraging our cost model, adapts to various network environments and automatically selects the optimal solution between \framework$_\text{VarPIR}$ and \framework$_\text{PlainDL}$. \framework achieves response times ranging from 0.12 to 42.6 seconds, depending on the available bandwidth and the specified security level.

\subsubsection{Impact of Value Size}\label{sec:eva_largevalue}
\Cref{fig:large_value} illustrates the changes in both online and offline execution times for each scheme as the byte length of each value varies. The dataset consists of $2^{20}$ key-value pairs, with fixed 8-byte keys and the low security level for \framework.
\begin{figure}[!t]
    \centering
    \includegraphics[width=1\linewidth]{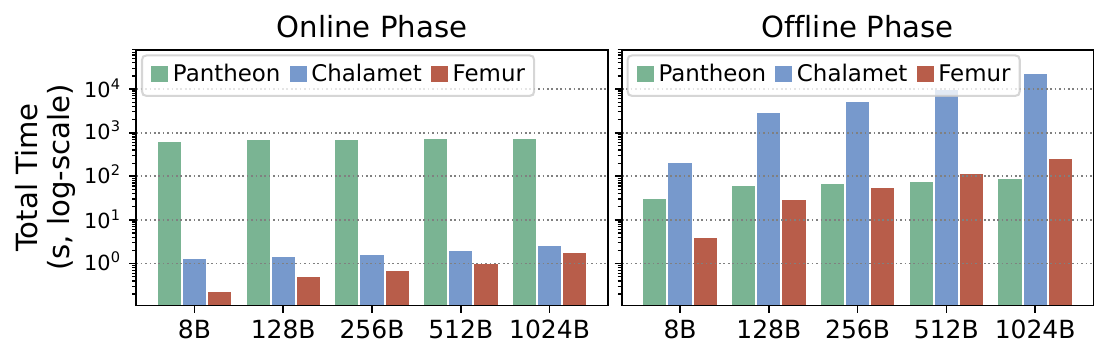}
    \caption{Online/Offline Execution Time (in seconds) at Various Value Size (in bytes)\textnormal{ -- Both graphs share the same y-axis.}}
    \label{fig:large_value}
\end{figure}

The online execution time of \framework increases as the value length grows, but it continues to show advantages. This increase is due to the larger total byte size, which leads to more encoded plaintexts and also slows down the plaintext encoding process during the offline phase. 
When the value length is 8 bytes, \framework uses the PlainDL scheme. As the values become larger, PlainDL's performance degrades rapidly with more transmitted data. Thus \framework switches to the VarPIR scheme at 128 byte values.

Chalamet’s online time increases slightly with the value length, but its offline time grows significantly, from 201.7 seconds to 6.02 hours. This is due to the rapid growth in computational cost caused by the larger key-value pair length, with most of the computation shifted to the offline phase. In contrast, Pantheon’s bottleneck lies in the equality check between the querying key and all keys. Since the key length remains constant, its online processing time increases only slightly, from 609 seconds to 729 seconds.

\subsubsection{Comparison of Client-Side Indexes}\label{sec:experiment_pgm}
To demonstrate the benefits of using PGM-indexes, we compare PGM-indexes and B+trees based on the length of the obfuscated ranges generated and the index size. 
The obfuscated range length directly impacts the data volume downloaded or processed on the server. The index size determines the amount of data required for clients to download from the server during the offline phase, which is also crucial especially in scenarios with update operations.
\begin{table}[!t]
    \centering
    \small
    \caption{Average Length of Obfuscated Ranges with Various Security Levels}
    \label{tab:btree_pgm_avg_range}
    \vspace{-0.5em}
    \begin{tabular}{c|cccc}
    \toprule
    \textbf{Distance} (\textbf{$t$}) &  \textbf{10}&\textbf{100} & \textbf{1000} & \textbf{10000} \\
    \midrule
    B+tree &  65541.46&65552.8 & 262088.2 & 2621408.6 \\
    \pgm &  2715.12&25740.2 & 256135.1 & 2559313.1 \\
    \midrule
    Ratio &  24.1394$\times$&2.5461$\times$ & 1.0233$\times$ & 1.0243$\times$ \\ 
    \bottomrule
    \end{tabular}
\end{table}

We first implement the noise generation algorithm described in~\Cref{sec:dp} on the B+tree and evaluate both indexes under various security levels. For each index, we process 10 million queries, converting querying keys into obfuscated ranges and measuring the average length of the resulting ranges. As shown in~\Cref{tab:btree_pgm_avg_range}, when the security requirements are low, the obfuscated ranges generated by the B+tree are 24.14$\times$ longer than those generated by the \pgm. Even as the security requirements increase, the \pgm consistently reduces approximately 60,000 injected noise. These results demonstrate that the \pgm, which operates at item-level granularity, is effective in minimizing injected noise while satisfying the same security requirements.

\begin{table}[!t]
    \centering
    \footnotesize 
    \caption{Index sizes of B+trees and PGM-indexes (in MiB)\textnormal{ -- The bolded numbers are the index sizes used in our previous experiments, indicating the network required to transmit the index.}}
    \label{tab:btree_pgm_size}
    \begin{tabular}{c c c c c c c c}
    \toprule
    \multirow{2}{*}[-0.5ex]{\textbf{Size}}&\multirow{2}{*}[-0.5ex]{\textbf{Dataset}}&\multirow{2}{*}[-0.5ex]{\textbf{B+tree}}&\multicolumn{5}{c}{\textbf{\pgm (With Various $\varepsilon_\text{data}$)}} \\
    \cmidrule(lr){4-8}
    &&&512&256&128&64&32\\
    \midrule
    \multirow{3}{*}{200 million}&Wiki&11.93 &0.11 & 0.17&0.29&0.58&1.31 \\
    &OSMC&11.93 & 0.65& 1.28& 2.55& \textbf{5.13}&10.3\\
    &Normal&11.93 &0.01& 0.01&0.01&0.01&0.02\\
    \midrule
    $2^{20}$& OSMC&0.13& 0.01 &0.01&0.02&\textbf{0.03}&0.07 \\
    $2^{22}$& OSMC&0.50& 0.01& 0.03& 0.06&\textbf{0.11}& 0.22\\
    $2^{24}$& OSMC&2.00 &0.06& 0.12& 0.24&\textbf{ 0.47}& 0.93\\
    $2^{26}$& OSMC&8.01& 0.25& 0.50& 0.98&\textbf{ 1.95}& 3.91 \\
    \bottomrule
    \end{tabular}
\end{table}
To evaluate the effects of data distribution and dataset size on index sizes, we construct both indexes using three datasets from the SOSD benchmark~\cite{marcus2020benchmarking}, each containing 200 million key-value pairs (two real-world datasets and one synthetic dataset following a normal distribution). Besides, we test index sizes on OSMC datasets of varying sizes. For the B+tree, only internal nodes are counted, and for the \pgm, we measure sizes across different values of $\varepsilon_{\text{data}}$.
As shown in~\Cref{tab:btree_pgm_size}, the \pgm is significantly smaller than the B+tree, even with small $\varepsilon_{\text{data}}$ error requirements and on the OSMC dataset that is less friendly to learned indexes.
Based on these results, we set the default value of $\varepsilon_\text{data}$ to 64 to balance the index size and the length of the predicted range. Meanwhile, $\varepsilon_\text{model}$ has a smaller impact on the index size, so we adopt the default value of 4 from the original paper~\cite{ferragina2020pgm}.

\subsection{Evaluation on Redis}\label{sec:eva_update}
To evaluate the performance and usability of \framework in real-world scenarios, we implemented \framework on Redis, a widely used key-value store, using Jedis (the Redis Java client)~\cite{jedis}. The server also continuously updates its key-value pairs while serving private queries from clients. 
During lookup and update operations, we use the Jedis interface to retrieve data within the specified range and write updated plaintext back to Redis. To ensure uninterrupted client queries during updates, we employ Multi-Version Concurrency Control (MVCC) to manage multiple Redis instances, as detailed in~\Cref{sec:update}.
We conduct experiments by executing 100 lookup queries on this Redis-based version of \framework under three scenarios: (1) no updates, (2) real-time updates for values with a 1-second interval between updates, and (3) periodic batch updates for keys with a 1-second delay between completing one update and initiating the next. The dataset size is $2^{24}$, and queries are performed at three security levels: low, medium, and high.
\begin{figure}[!t]
    \centering
    \includegraphics[width=1\linewidth]{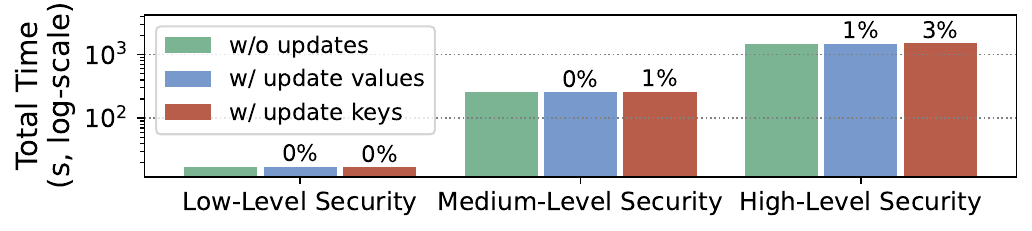}
    \caption{Online Execution Time for 100 Queries on Redis and with Updates\textnormal{ -- Each number above the bar is the ratio of the additional time caused by updates to the time without updates.}}
    \label{fig:redis}
\end{figure}

As shown in~\Cref{fig:redis}, updates cause at most a 3\% increase in total query time. This minor delay arises because some queries require fetching the updated \pgm after a version switch, ensuring that all lookup results remain up-to-date. We also track the number of update operations performed during the execution of 100 queries. In the high-security evaluation, the server completes 1,120 real-time value updates and 24 periodic batch updates for keys/databases. The real-time updates take only 70 ms each, including reading plaintext from Redis, updating it, and writing it back to Redis. A periodic batch update takes approximately 60 seconds, however, we use MVCC to further reduce the impact of update operations by performing most tasks asynchronously in background threads. Compared to Chalamet, which requires 1.5 hours to reinitialize a database of the same size, \framework significantly reduces update cycles, making it more practical for real-world scenarios.

\section{Related Work}
\para{Relaxed Security in Query Processing.}
Several works~\cite{bater2018shrinkwrap,qiu2023doquet,qin2022adore,zhang2023longshot,wang2024pripl,seeman2023privately,zhang2023dprovdb,fu2023dp,dong2023better,wang2025jodes,luo2023secure} have applied differential privacy (DP) to data management to provide relaxed security models with theoretical guarantees for secure query processing, thus enhancing performance. Specifically, Shrinkwrap~\cite{bater2018shrinkwrap} introduces DP-based noise to intermediate results, concealing the true size of the intermediate data and improving overall performance by eliminating the need for worst-case padding after each operator's execution. Adore~\cite{qin2022adore} and Doquet~\cite{qiu2023doquet} address access pattern leakage during relational operations and join queries using differentially oblivious operators and private data structures. Longshot~\cite{zhang2023longshot} tackles the challenges of indexing a growing database by combining DP with secure multiparty computation (MPC).
While prior work primarily targets scenarios involving online analytical processing (OLAP) or encrypted data, \framework enables private queries on public key-value stores, offering relaxed security with theoretical guarantees, ensuring that each query is indistinguishable from its neighbors.

\para{Privacy Enhancements in DBMS.}
In recent years, significant efforts have focused on enabling clients to outsource private data while ensuring secure query processing~\cite{popa2011cryptdb,antonopoulos2020azure,arasu2013orthogonal,eskandarian2017oblidb,priebe2018enclavedb,tu2013processing,sha2024object}. These approaches typically involve encrypting data, storing it on a cloud server, and using techniques such as order-preserving encryption~\cite{agrawal2004order}, homomorphic encryption~\cite{wong2014secure,tu2013processing,poddar2016arx}, searchable encryption~\cite{cash2014dynamic} and Trusted Execution Environments (TEE)~\cite{costan2016intel,bailleu2019speicher,sun2021building,wang2022operon,battiston2024duckdb,sha2023tee,zou2024salus} to achieve this goal. In contrast, \framework is designed for public datasets, which focuses on query privacy. Many works align with our goal~\cite{angel2018pir,menon2022spiral,melchor2016xpir,ahmad2021addra,mittal2011pir,zhou2024piano,patel2023don,ali2021communication}, mainly based on homomorphic encryption. For example, Pantheon~\cite{ahmad2022pantheon} and Constant-weight PIR~\cite{mahdavi2022constant} implement equality-checking operators and invoke traditional PIR schemes for homomorphic operations. While our framework supports similar homomorphic operations to return target key-value pairs, it also offers the flexibility to utilize other schemes, such as direct downloads when bandwidth is sufficient. Moreover, existing PIR schemes typically require all data to participate in computation or communication, resulting in poor scalability and limited support for large datasets. However, our scheme allows for flexibility in adjusting the range of data participation in computations without requiring re-encoding. This works seamlessly with our definition of distance-based indistinguishability. By offering users the ability to relax security and select different levels of security guarantees based on their needs, \framework significantly accelerates query response, making \framework a more practical solution.

\section{Conclusion}


We present \framework, a framework that enables users to perform secure queries on public key-value stores, while empowering users to tailor privacy protections to their specific needs. 
By employing the novel concept of distance-based indistinguishability and an adaptive retrieval mechanism supporting both direct downloads and an enhanced PIR scheme, \framework achieves a fine balance between privacy and performance. Our evaluations confirm that \framework efficiently supports a wide range of security configurations while maintaining practical query response times, even on large datasets.

\begin{acks}
This work was partially supported by the Shanghai Qi Zhi Institute Innovation Program (SQZ202406 \& SQZ202314) and the National Social Science Foundation of China (Grant No. 22 \& ZD147).
\end{acks}

\newpage

\balance
\bibliographystyle{ACM-Reference-Format}
\bibliography{sample-base}

\end{document}